\documentclass[11pt]{article}

\usepackage{amsmath}
\usepackage{amssymb}
\usepackage{graphicx}
\usepackage{float}
\usepackage{caption}
\usepackage{amsthm}
\usepackage{color}
\definecolor{green}{rgb}{0,0.5977,0}
\usepackage[linesnumbered,algoruled,boxed,lined]{algorithm2e}

\newcommand{\rr}{\mathbb R}

\newcommand{\abs}[1]{\left|{#1}\right|}

\newcommand{\suchthat}{\ | \ }

\newcommand{\genseq}[3]{{#1}_1 {#3} {#1}_2 {#3} \dots {#3} {#1}_{#2}}
\newcommand{\seq}[2]{\genseq{#1}{#2}{,}}

\newcommand{\txt}[1]{\text{#1}}

\newcommand{\stext}[1]{\ \ \ \ \ \text{(#1)}}

\newcommand{\ipnc}[3]{\begin{figure}[H]\begin{center}\includegraphics[scale = {#1}]{#2.pdf}\caption{#3}\end{center}\end{figure}}

\makeatletter 
\g@addto@macro{\@algocf@init}{\SetKwInOut{Parameter}{Parameters}} 
\makeatother

\newcommand{\rp}[1]{\rr\text{P}^{#1}}

\theoremstyle{plain}
\newtheorem{theorem}{Theorem}
\newtheorem{lemma}[theorem]{Lemma}

\theoremstyle{definition}

\numberwithin{theorem}{section}

\usepackage{authblk}
\usepackage{enumitem}

\title{Topological Universality of the Art Gallery Problem}
\author[1]{Jack Stade}
\author[2]{Jamie Tucker-Foltz}
\affil[1]{University of Cambridge. \texttt{jack.stade@gmail.com}}
\affil[2]{Harvard University. \texttt{jtuckerfoltz@gmail.com}}
\begin{document}
\maketitle
\begin{abstract}
    We prove that any compact semi-algebraic set is homeomorphic to the solution space of some art gallery problem. Previous works have established similar universality theorems, but holding only up to homotopy equivalence, rather than homeomorphism, and prior to this work, the existence of art galleries even for simple spaces such as the M\"obius strip or the three-holed torus were unknown. Our construction relies on an elegant and versatile gadget to copy guard positions with minimal overhead. It is simpler than previous constructions, consisting of a single rectangular room with convex slits cut out from the edges. We show that both the orientable and non-orientable surfaces of genus $n$ admit galleries with only $O(n)$ vertices.
\end{abstract}

\section{Introduction}\label{secIntro}

An instance of the \emph{art gallery problem} consists of a polygon $P$ (which we refer to as the \emph{art gallery}), and the objective is to find a finite set of points $G \subseteq P$ (the \emph{guards}) of minimal cardinality such that every point in $P$ is \emph{visible} to some guard, meaning that the line segment between the point and the guard is contained within $P$. This problem was first introduced by Viktor Klee in 1973 and has a long history. In 1986, Lee and Lin \cite{NPHardness} showed that the decision problem of determining whether there exists a configuration with at most $k$ guards is NP-hard, but the problem is not known to be in NP since there is no obvious succinct way to represent the guards' coordinates in a solution (they might have to be irrational, even if all polygon vertex coordinates are rational \cite{IrrationalGuards}). In 2018, Abrahamsen, Adamaszek, and Miltzow \cite{ExistsRHardness} showed that the problem is $\exists\rr$-complete, mostly settling the question of complexity (modulo the longstanding conjectures that $\txt{NP} \subsetneq \exists\rr \subsetneq \txt{PSPACE}$). Approximation algorithms and lower bounds have also been studied; see Bonnet and Miltzow \cite{ApproximationAlgorithm} for a recent overview. 

Aside from complexity aspects, a parallel line of inquiry concerns the topology of the space of solutions. Supposing that gallery $P$ requires exactly $k$ guards at minimum, we let
\[V(P) := \{G \subseteq P \suchthat \abs{G} = k \txt{ and every $p \in P$ is visible to some } g \in G\}.\]
The set $V(P)$ consists of unordered sets of points of cardinality $k$, which can be turned into a metric space using the Hausdorff distance
\[d_H(G_0, G_1) := \max_{i \in \{0, 1\}}\max_{g \in G_i}\min_{g' \in G_{1 - i}}d(g, g'),\]
where $d$ denotes the Euclidean distance on $\rr^2$. Thus, for any art gallery $P$, $V(P)$ is a topological space, so a natural question is, what kinds of topological spaces can occur?

More than a mere mathematical curiosity, this question is relevant to complexity theory due to the connections between $\exists\rr$-hardness and topological universality. The most famous example is Mn\"ev's theorem \cite{Mnev, Richter1995} that any semi-algebraic set is stably-equivalent to the space of point configurations of some order type, inspiring an eventual proof that order type realizability is $\exists\rr$-complete \cite{OrderTypeComplexity}. A \emph{semi-algebraic} set is a finite union $\bigcup_{i = 1}^m S_i$, where each $S_i$ is defined to be the set of points $(\seq{x}{n}) \in \rr^n$ satisfying a finite number of constraints of the form $P(\seq{x}{n}) \geq 0$ or $P(\seq{x}{n}) > 0$, where $P$ is a polynomial. The canonical complete problem for $\exists\rr$ is called ETR (Existential Theory of the Reals), which asks whether a given semi-algebraic set is nonempty. As a consequence of their reduction from ETR to the art gallery problem, Abrahamsen et al.\txt{} \cite{ExistsRHardness} show, for any compact semi-algebraic set $S$, how to construct an art gallery $P$ such that $V(P)$ surjects continuously onto $S$. However, they do not show that the mapping is injective, so this fails to establish universality.

In a recent paper, Bertschinger, El Maalouly, Miltzow, Schnider, and Weber \cite{SOSAMain} show that any compact semi-algebraic set $S$ is homotopy-equivalent to $V(P)$ for some polygon $P$. They leave as an open question whether $P$ can be constructed so that $V(P)$ is not just homotopy-equivalent, but homeomorphic to $S$. Only the following list of spaces are shown to be captured up to homeomorphism:

\begin{itemize}
	\item $k$-clover (obtained by joining $k$ circles at a single point)
	\item $k$-chain (obtained by connecting $k$ circles in a path with $k - 1$ disjoint line segments)
	\item $4k$-necklace (obtained by connecting $4k$ circles in a cycle with $4k$ disjoint line segments)
	\item $k$-sphere
	\item The torus
	\item The 2-holed torus
\end{itemize}

The constructions for these spaces are all based on simple galleries with solution spaces homeomorphic to a circle or an interval, which are combined to give Cartesian products and then given simple additional constraints. However, using these methods it is difficult to obtain more general spaces because the geometry significantly limits the types of constraints that can be used. Thus, prior to this work, homeomorphism universality was unknown even for closed surfaces. Galleries for the real projective plane, Klein bottle, and M\"obius strip were explicitly left as open questions. 

In this work, we settle the question of homeomorphism universality in the affirmative:

\begin{theorem}\label{thmMain}
	For every compact semi-algebraic set $S \neq \emptyset$, there exists a polygon $P$ such that $S$ is homeomorphic to $V(P)$.
\end{theorem}

In addition to yielding a strictly stronger result than that of Bertschinger et al.\txt{} \cite{SOSAMain}, our construction is structurally simpler: the art gallery always consists of a single rectangular room with two types of convex slits repeatedly cut out of the edges. As an example, we explicitly draw an art gallery for the M\"obius strip (to the best resolution that can reasonably fit on a page) in Figure \ref{figMobiusGallery}.

The key ingredient in our approach is a novel form of \emph{copying gadget} which enforces a constraint of the form $x_i = x_j$, thereby requiring multiple guards to represent the same underlying variable. Gadgets with similar functions can also be found in Abrahamsen et al.\txt{} \cite{ExistsRHardness}. Our gadget takes a geometrically different form, and is an improvement for the following reasons:
\begin{itemize}
	\item It is more versatile. In particular the polygon $P$ does not need to intersect the convex hull of the segments being copied. This makes it possible for the gallery to only have a single ``room.''
	\item It works in more general contexts. The previous gadget uses constraints created by the need for guards to see every point on the interior of the polygon $P$. Our gadget works in the variant of the problem where guards only need to see the boundary of the polygon. Stade \txt{} \cite{stade2022complexity} uses this gadget in a forthcoming paper proving that this variant is also $\exists\rr$-hard.
\end{itemize}

Our general construction, which we present in Section \ref{secGeneralConstruction}, does not yield an obvious bound on the size of $P$ (number of vertices) as a function of the complexity of the space $S$. In Section \ref{secEfficientConstruction}, we refine our construction to show that both the orientable and non-orientable surfaces of genus $n$ can be captured by art galleries with only $O(n)$ vertices.

\section{General construction}\label{secGeneralConstruction}

As in Bertschinger et al.\txt{} \cite{SOSAMain}, the starting point for our reduction is a theorem by Hironaka \cite{Hironaka} that any compact semi-algebraic set can be triangulated as a cubical complex. We begin by reviewing this result. Next, we introduce our variable/copying gadgets and establish its key properties. Finally, we show how to combine the gadgets to construct the art gallery. We illustrate our construction using the M\"obius strip as a running example.

\subsection{Hironaka's theorem}\label{subHironaka}

An {\it abstract cubical complex} is a subset $K$ of the set
\[I^n=\left\{\left\{\mathbf{x}\in [0, 1]^n \suchthat x_{i_1}=c_1, \dots x_{i_k}=c_k\right\} \suchthat 1\le i_1<\dots<i_k\le n, \left(c_1,\dots, c_k\right) \in \{0, 1\}^k \right\}\]
of faces of an $n$ dimensional hypercube $[0, 1]^n$ such that if $a\in K$ and $b\subseteq a$, then $b\in K$. We write $|K|:=\bigcup K \subseteq [0, 1]^n$ for the union of faces, called the {\it geometric realization} of $K$.

\begin{theorem}[Bertschinger et al.\txt{} \cite{SOSAMain}, Lemma 3]\label{thmCubeComplexUniversality}
	Any compact semi-algebraic set is homeomorphic to the geometric realization of an abstract cubical complex.
\end{theorem}

\begin{proof}[Proof sketch.]
	This fact is well-known so we do not give a complete proof here. Using Hironaka's theorem \cite{Hironaka}, it is possible to show that any compact semi-algebraic set is homeomorphic to the geometric realization of an abstract simplicial complex (see Hoffman \cite{SemiAlgebraicTriangulation}). Additionally, it can be shown that any abstract simplicial complex has an abstract cubic complex with a homeomorphic geometric realization (see Blass and Wlodzimierz \cite[Theorem 1.1]{Cubic}).
\end{proof}

For example, the circle $S^1$ can be represented as the abstract cubical complex
\[\{\{\mathbf{x} \in [0, 1]^2 \suchthat x_1 = 0\}, \{\mathbf{x} \in [0, 1]^2 \suchthat x_1 = 1\}, \{\mathbf{x} \in [0, 1]^2 \suchthat x_2 = 0\}, \{\mathbf{x} \in [0, 1]^2 \suchthat x_2 = 1\}\},\]
in which case its geometric realization is the boundary of the unit square in $\rr^2$. The M\"obius strip can be realized as a subset of the 4-dimensional hypercube as in Figure \ref{figMobiusHypercube}.\

\ipnc{0.3}{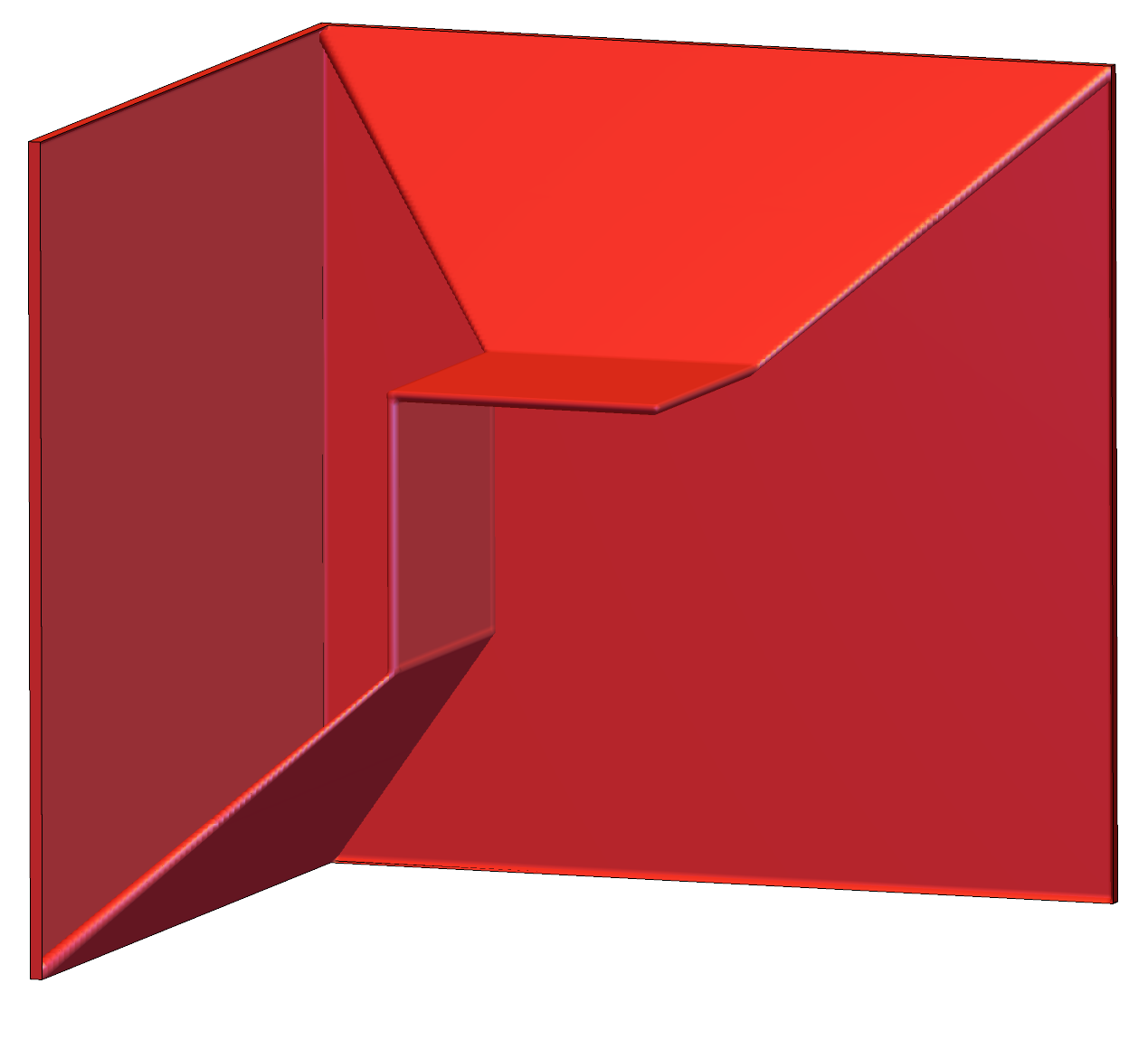}{\label{figMobiusHypercube}A 4-dimensional geometric realization of a cubical complex homeomorphic to a M\"obius strip.}

\subsection{Variable gadgets}\label{subVarGadget}

The art gallery we construct will be an axis-aligned rectangle with several slits removed from around the border. We refer to certain slits and combinations of slits as \emph{gadgets}. Our construction involves three kinds of gadgets: \emph{variable gadgets}, \emph{copying gadgets}, and \emph{clause gadgets}.

We begin by establishing the key properties of the variable gadgets, which also appear in Bertschinger et al.\txt{} \cite{SOSAMain}. A variable gadget consists of one slit in the left wall and two in the right wall, as depicted in Figure \ref{figVarGadget}.

\ipnc{.218}{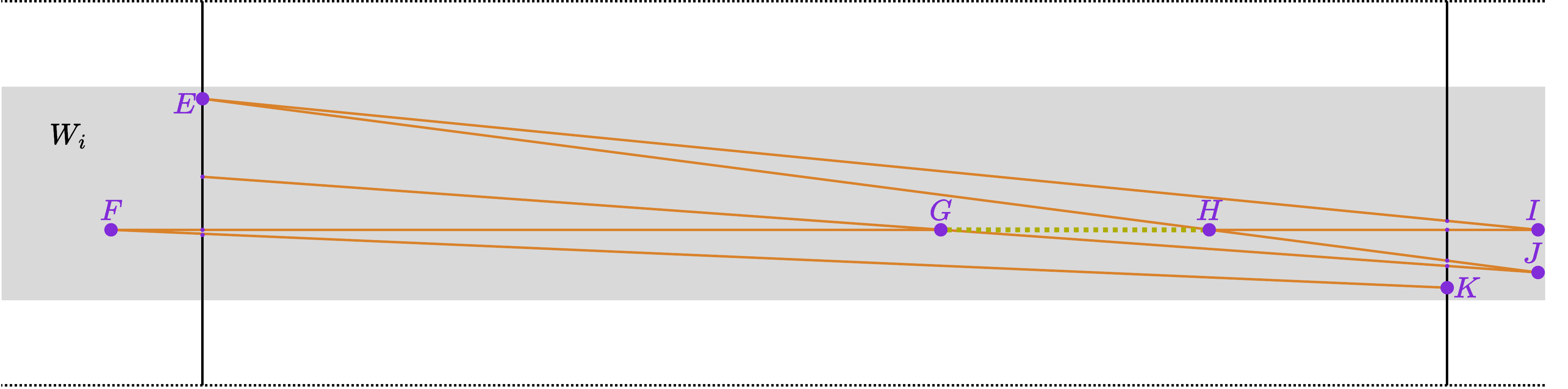}{\label{figVarGadget}A variable gadget enforcing that a guard must be placed on the dashed gold line segment $GH$. The boundary of this section of the art gallery consists of the entire outer profile (continuing above and below the dashed lines at the top and bottom), involving line segments of all colors. No other gadgets may intersect the shaded rectangular region $W_i$.}

\begin{lemma}\label{lemVarGadget}
	If slits are drawn as in Figure \ref{figVarGadget}, with no other slits in the region $W_i$, then at least one guard must be placed within $W_i$. If there is only one guard in that region, it must be placed on line segment $GH$. Furthermore, the slits can be drawn so that the height of $W_i$ is arbitrarily small.
\end{lemma}

\begin{proof}
	To see point $F$, we must have at least one guard placed within the triangle $FIK$ (which also requires there to be some guard in $W_i$). To see point $I$, we must have at least one guard placed within the triangle $EIF$, so if there only one guard in $W_i$, it must be on the line segment $FI$. Since it must also see the point $J$, this further restricts the guard to the segment $GH$. The entire figure can be scaled in the $y$ direction to give $W_i$ arbitrarily small height without changing the dimensions of the segment $GH$.
\end{proof}

We refer to the line segments $GH$ from Lemma \ref{lemVarGadget} as \emph{guard segments}, which we will number $1, 2, \dots, n$, and we refer to $W_i$ as the \emph{guard region} of guard segment $i$. Since guard regions can be made arbitrarily small, we may arrange these guard segments any way we wish within the rectangular gallery so long as no two guard segments share the same $y$-coordinate.

\subsection{Copying gadgets}\label{subCopyGadget}

We next introduce our copying gadgets, which are based on the following geometric observation.

\begin{lemma}\label{lemInversionLinearity}
	Let $L_1$ and $L_2$ be two distinct parallel lines in the plane, and let $Z$ be a point strictly between $L_1$ and $L_2$. Let $f: L_1 \to L_2$ be the map defined via inversion through $Z$, i.e., for $p \in L_1$, $f(p)$ is the intersection of the line $pZ$ with $L_2$. Then $f$ is linear (meaning it preserves ratios of distances between points on $L_1$ and the corresponding points on $L_2$).
\end{lemma}

\ipnc{0.66}{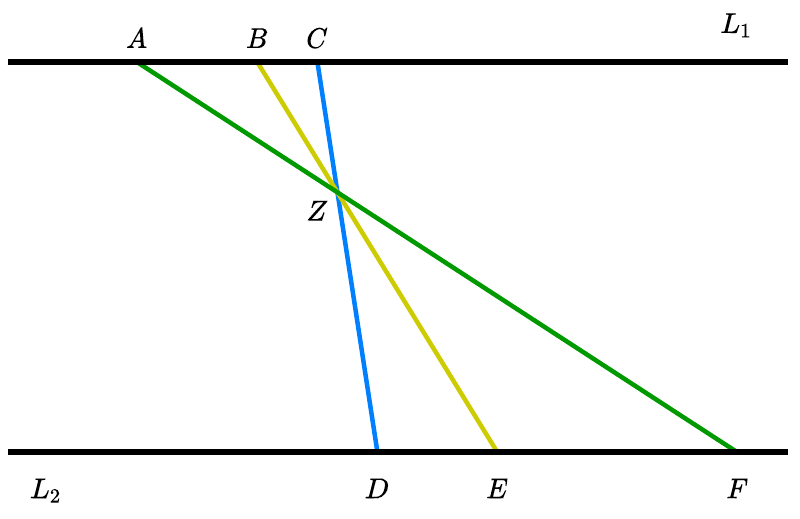}{\label{figSimilarTriangles}Illustration of the geometry underlying the copying gadget.}

\begin{proof}
	Assume $L_1$ and $L_2$ are drawn aligned with the $x$-axis as in Figure \ref{figSimilarTriangles}. Let $A$, $B$, and $C$ be three arbitrary distinct points on $L_1$, and let $F$, $E$, and $D$ be their respective images under $f$. It is easy to see that $(AZB, FZE)$ and $(AZC, FZD)$ are pairs of similar triangles. This means that the ratios of lengths of $AB$ to $FE$ and $AC$ to $FD$ are equal, that is $\frac{B_x-A_x}{F_x-E_x}=\frac{C_x-A_x}{F_x-D_x}$. Rearranging, we have $\frac{B_x-A_x}{C_x-A_x}=\frac{F_x-E_x}{F_x-D_x}$.
\end{proof}

We now discuss the key properties of the copying gadgets, which are depicted in Figure \ref{figCopyGadget}.

\begin{lemma}\label{lemCopyGadget}
	Suppose guard segments $GH$ and $NO$ are horizontally aligned and parallel as in Figure \ref{figCopyGadget}, with disjoint slits $CBAD$ and $SUVT$ lying outside of the guard regions such that $AB$, $UV$, $GH$ and $NO$ are all parallel. Furthermore, suppose $CBAD$ is chosen so the triples $(A, C, H)$, $(B, C, G)$, $(A, D, O)$ and $(B, D, N)$ are each colinear, and $SUVT$ is chosen similarly, as in the figure. Then any valid configuration of guards with one guard in each guard region must have two guards sharing an $x$-coordinate, one on $GH$ and one on $NO$. Furthermore, if the left wall is moved sufficiently far to the left, $CBAD$ can be placed anywhere above $W_i$ not in a guard region, $SUVT$ can be placed anywhere below $W_j$ not in a guard region, and both of these slits can be made arbitrarily small.
\end{lemma}

\begin{proof}
	Let $\alpha$ be the location of the guard on $GH$, and let $\beta$ be the location of guard on $NO$; our objective is to show that, in any valid configuration, $\alpha_x = \beta_x$. We define the following four additional points:
	\begin{itemize}
		\item $\overline{\alpha} :=$ the intersection of line segment $AB$ with line $\alpha C$
		\item $\underline{\alpha} :=$ the intersection of line segment $UV$ with line $\alpha S$
		\item $\overline{\beta} :=$ the intersection of line segment $AB$ with line $\beta D$
		\item $\underline{\beta} :=$ the intersection of line segment $UV$ with line $\beta T$
	\end{itemize}
	
	\ipnc{.212}{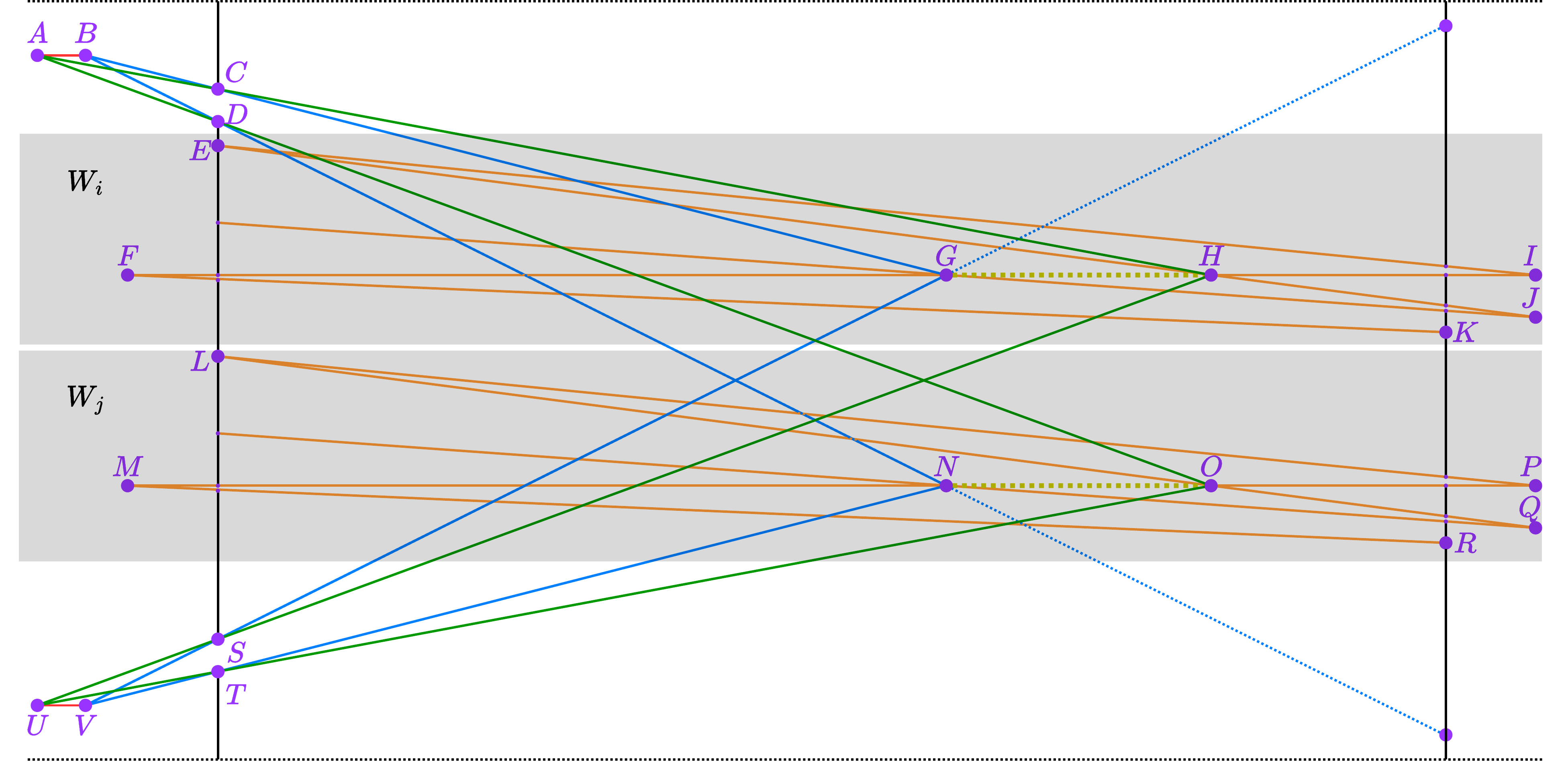}{\label{figCopyGadget}A copying gadget (and two variable gadgets). As in Figure \ref{figVarGadget}, the boundary consists of the entire outer profile. The copying gadgets force the pair of guards on segments $GH$ and $NO$ to have the same $x$-coordinate. Additional variable and copying gadgets may be placed between $W_i$ and $W_j$, which would be easier to draw if $W_i$ and $W_j$ were made narrower; they are drawn very large in this figure for the purpose of illustration.}
	
	Due to the obstructions by points $C$ and $D$, the guard at $\alpha$ can only see $AB$ to the left of $\overline{\alpha}$, and the guard at $\beta$ can only see $AB$ to the right of $\overline{\beta}$. If $\overline{\alpha}_x < \overline{\beta}_x$, the line segment $\overline{\alpha}\overline{\beta}$ will therefore not be seen by any guard, so we must have that $\overline{\alpha}_x \geq \overline{\beta}_x$. Similarly, the guard at $\alpha$ can only see $UV$ to the right of $\underline{\alpha}$, and the guard at $\beta$ can only see $UV$ to the left of $\underline{\beta}$, so we must have that $\underline{\alpha}_x \leq \underline{\beta}_x$. Thus, in any valid configuration of the two guards,
	{\allowdisplaybreaks\begin{align*}
		\frac{\alpha_x - G_x}{H_x - G_x} &= \frac{B_x - \overline{\alpha}_x}{B_x - A_x}\stext{by Lemma \ref{lemInversionLinearity} with $Z := C$}\\
		&\leq \frac{B_x - \overline{\beta}_x}{B_x - A_x}\\
		&= \frac{\beta_x - N_x}{O_x - N_x}\stext{by Lemma \ref{lemInversionLinearity} with $Z := D$}\\
		&= \frac{V_x - \underline{\beta}_x}{V_x - U_x}\stext{by Lemma \ref{lemInversionLinearity} with $Z := T$}\\
		&\leq \frac{V_x - \underline{\alpha}_x}{V_x - U_x}\\
		&= \frac{\alpha_x - G_x}{H_x - G_x}\stext{by Lemma \ref{lemInversionLinearity} with $Z := S$}.
	\end{align*}}
	Multiplying through by $(H_x - G_x)$ and using the fact that $N_x = G_x$ and $O_x = H_x$, we have
	\[\alpha_x - G_x \leq \beta_x - G_x \leq \alpha_x - G_x.\]
	Thus, $\alpha_x = \beta_x$ in any valid configuration (and it is obvious that such a configuration is indeed valid).
	
	For the final claim, observe that we can fully define the position of the top slit $CBAD$ as follows. Place $AB$ anywhere outside the art gallery, not in any guard region. Then move the left wall (including all slits) sufficiently to the left so that the $y$-coordinates of points $C$ and $D$ are sufficiently close to $A$ and $B$ so that they do not lie in any guard region either. Existing slits will have to get stretched in this process, which is fine because this only makes them smaller. Thus, it is possible to make the top slit $CBAD$ arbitrarily small and place it anywhere, so long as it is above $W_i$ and not in any other guard region. We can apply a symmetric procedure for $SUVT$.
\end{proof}

\subsection{Constructing the art gallery}

We are now ready to prove homeomorphism universality. Throughout the proof, as an example, we implement the various steps the construction for the M\"obius strip.

\begin{proof}[Proof of Theorem \ref{thmMain}]
	By Theorem \ref{thmCubeComplexUniversality}, we may assume without loss of generality that $S$ is the geometric realization of a nonempty cubical complex. This means that $S$ can be described as a subset of $[0, 1]^n$ whose coordinates $\seq{x}{n}$ satisfy some disjunctive normal form (DNF) formula
	\[\phi_S \equiv \bigvee_{j = 1}^{m'} C'_j,\]
	where each $C'_i$ is a conjunction of constraints that certain $x_i$ variables take values 0 and 1. The formula will always be in DNF: each $C'_i$ corresponds to a face of the cubical complex, so $\phi_S$ describes the set of points lying in at least one such face. For example, the formula for the M\"obius strip $M$ depicted in Figure \ref{figMobiusHypercube} has six clauses, one for each of the six 2-dimensional faces. To write an explicit formula $\phi_M$, we choose a coordinate system where $x_4 = 0$ corresponds to the points on the outer shell and enumerate the faces as clauses by traversing the strip starting from the left-most depicted face and proceeding next toward the back face:
	\begin{align*}
		\phi_M \equiv& (x_2 = 0 \wedge x_4 = 0) \vee (x_1 = 0 \wedge x_4 = 0) \vee (x_1 = 0 \wedge x_3 = 1)\\
		\vee& (x_3 = 1 \wedge x_4 = 1) \vee (x_2 = 0 \wedge x_4 = 1) \vee (x_2 = 0 \wedge x_3 = 0)
	\end{align*}
	
	We next rewrite $\phi_S$ in conjunctive normal form (CNF),
	\[\phi_S \equiv \bigwedge_{j = 1}^m C_i,\]
	where each $C_i$ is a disjunction of constraints, i.e.,
	\[C_i \equiv (x_{i_{j, 1}} = c_{j, 1}) \vee (x_{i_{j, 2}} = c_{j, 2}) \vee \dots \vee (x_{i_{j, \ell}} = c_{j, \ell}),\]
	with each $i_{j, k} \in \{1, 2, \dots n\}$ and $c_{j, k} \in \{0, 1\}$. The transformation from DNF to CNF is standard, and can be accomplished by enumerating all tuples of constraints that take one constraint from each clause. For example, since $\phi_M$ has 6 clauses in DNF, each of size 2, the translation to CNF produces $2^6 = 64$ clauses, each of size 6. However, after eliminating redundancies, we can simplify $\phi_M$ to
	\begin{align}
		\phi_M \equiv& (x_1 = 0 \vee x_2 = 0 \vee x_3 = 1)\nonumber\\
		\wedge& (x_2 = 0 \vee x_3 = 1 \vee x_4 = 0)\nonumber\\
		\wedge& (x_1 = 0 \vee x_2 = 0 \vee x_4 = 1) \label{equMobiusCNF}\\
		\wedge& (x_3 = 0 \vee x_3 = 1 \vee x_4 = 0 \vee x_4 = 1)\nonumber\\
		\wedge& (x_1 = 0 \vee x_3 = 0 \vee x_4 = 0 \vee x_4 = 1)\nonumber.
	\end{align}
	
	\ipnc{0.79}{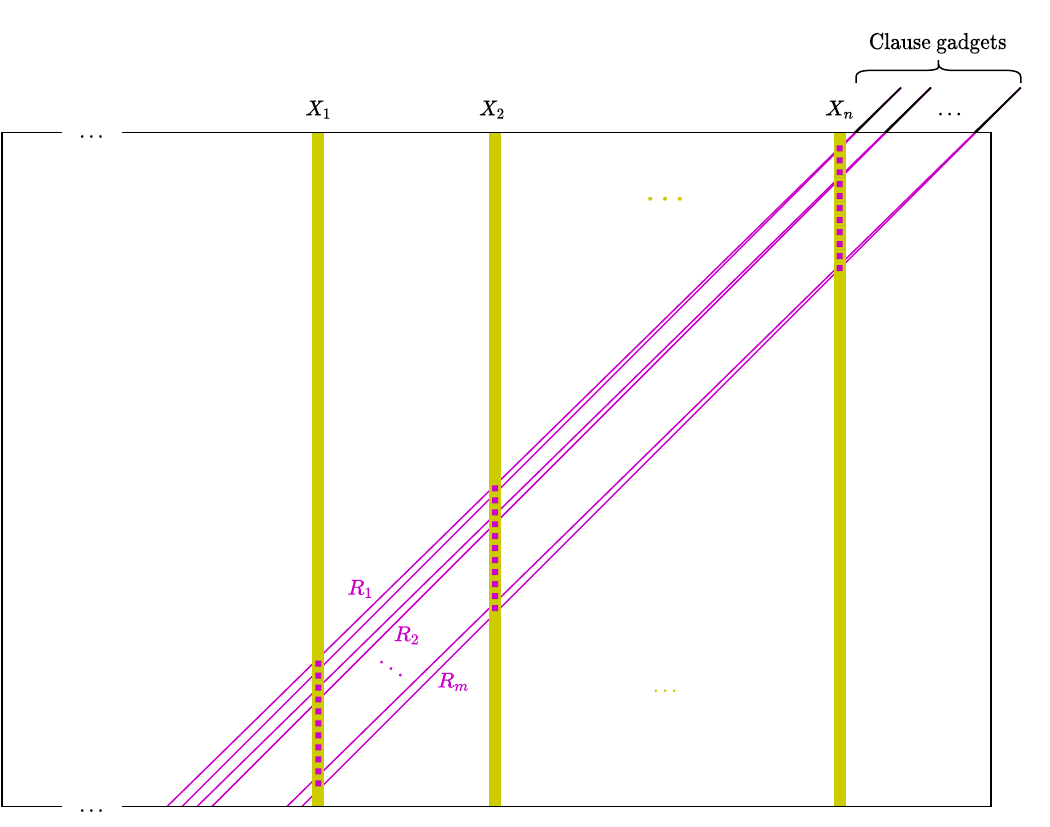}{\label{figSchematic}Illustration of the first step of the general construction (making clause gadgets), starting from a CNF formula with $n$ variables and $m$ clauses. The full art gallery is shown in more detail for the M\"obius strip in Figure \ref{figMobiusGallery}.}
	
	Starting from a rectangular art gallery, we make a narrow diagonal slit for each clause $j$ in the top-right corner such that the regions $R_j$ of the gallery that can see to the end of each slit extend downward to the left and do not overlap, as shown in Figure \ref{figSchematic}. For each $i \in \{1, 2, 3, \dots, n\}$ we define $X_i$ to be a tall, skinny, axis-aligned rectangular region such that the convex hulls of the sets $X_i \cap (R_1 \cup R_2 \cup \dots \cup R_m)$ do not overlap in $y$-coordinates (in terms of Figure \ref{figSchematic}, the pink dashed lines must not overlap in $y$-coordinates). It is always possible to guarantee this non-overlapping property by making the gallery sufficiently tall and/or making the clauses sufficiently close together.
	
	In every clause $j$, for every constraint $x_{i_{j, k}} = c_{j, k}$, we add a variable gadget in the left and right sides of the rectangle enforcing the constraint that there is a guard on a guard segment spanning the width of $X_{i_{j, k}}$ whose endpoint lies within $R_j$. If $c_{j, k} = 0$, this will be the left endpoint, and if $c_{j, k} = 1$, this will be the right endpoint. We then add copying gadgets to enforce that all guard segments placed in $X_i$ for the same $i$ must have the same $x$-coordinate. By Lemmas \ref{lemVarGadget} and \ref{lemCopyGadget}, this is possible by shrinking the variable and copying gadgets and moving the left wall sufficiently far away.
	
	\ipnc{0.79}{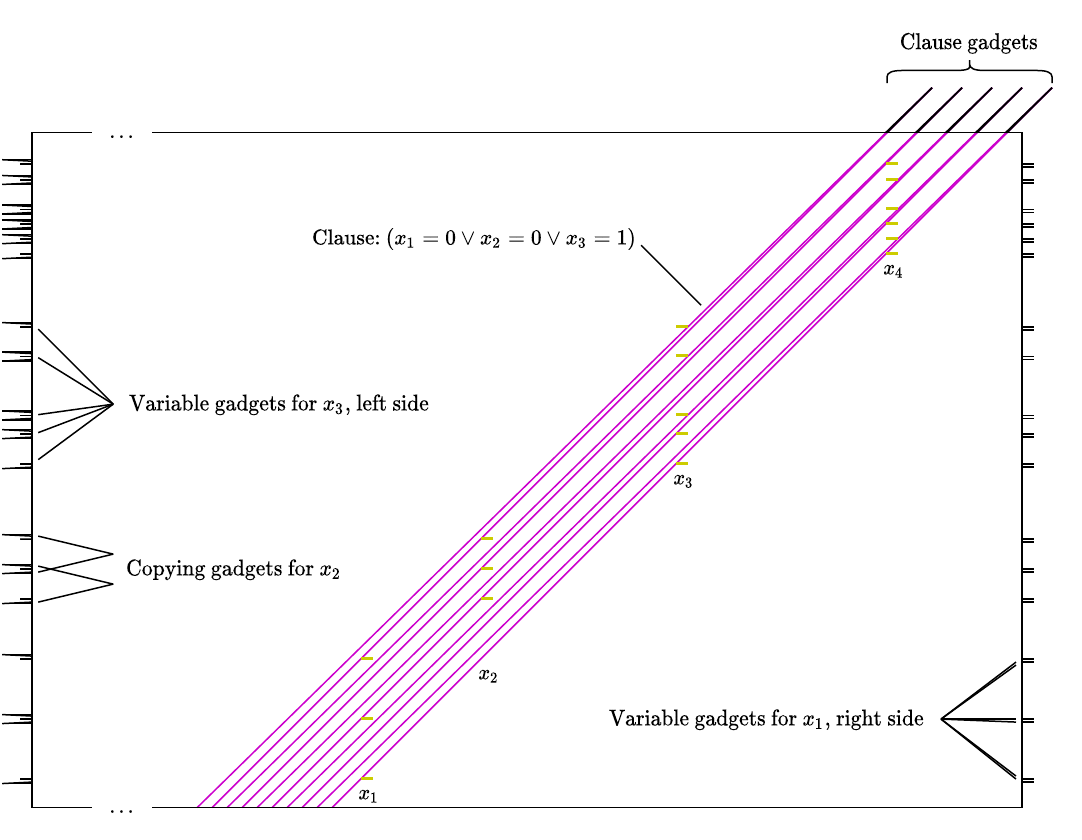}{\label{figMobiusGallery}An art gallery whose solution space is homeomorphic to a M\"obius strip, using the formula $\phi_M$ from (\ref{equMobiusCNF}). The horizontal gold dashes above each variable are the line segments on which guards must walk, as in Figure \ref{figCopyGadget}. The pink diagonal lines depict the regions $R_j$ that can see to the end of each clause gadget slit; each of these regions must contain at least one guard. The left wall must be placed sufficiently far to the left to ensure that none of the variable/copy gadgets interfere with each other. As drawn, this particular art gallery has 183 vertices.}
	
	Letting $P_0$ be the polygon $P$ without the clause gadget slits, we have that $V(P_0)$ consists of all solutions with one guard on each guard segment, with guards within the same $X_i$ placed at the same $x$-coordinate. There is thus a natural homeomorphism $h: V(P_0) \to [0, 1]^n$, and from the way the clause gadgets were constructed, it clearly follows that $V(P)$ consists of all solutions in which $\phi_S$ is satisfied under $h$. Thus, $V(P)$ is homeomorphic to $S$.
\end{proof}

\section{Efficient construction for closed surfaces}\label{secEfficientConstruction}

We have argued that our universality construction is qualitatively simpler than its predecessors. In this section, we show how our technique can be used to produce quantitatively simple galleries, in terms of the number of vertices of the polygon. This is not apparent a priori: even if a space can be triangulated as a cubical complex with relatively few faces, the conversion from DNF to CNF can exponentially blow up the number of clauses, and thus the size of the art gallery. Here we prove that an important class of topological spaces, namely the closed surfaces, can occur as solutions to art galleries with linearly many vertices.

\begin{theorem}
	There are polygons $P_g$, $Q_g$ with $O(g)$ vertices such that $V(P_g)$ is homeomorphic to the closed orientable surface of genus $g$ and $V(Q_g)$ is homeomorphic to the closed non-orientable surface of genus $g$.
\end{theorem}

\begin{proof}
	We know by Theorem \ref{thmMain} that such polygons exist for the finite number of cases where $g \in \{0, 1\}$, so it is sufficient to construct these polygons only for $g \geq 2$.
	
	It is well known that, for $g \geq 2$, the orientable surface of genus $g$ can be obtained as the connected sum of $g$ copies of a torus, $T^2 \# T^2 \# \dots \# T^2$, while the non-orientable surface of genus $g$ can be obtained as the connected sum of $g$ copies of the real projective plane $\rr\mathbb{P}^2 \# \rr\mathbb{P}^2 \# \dots \# \rr\mathbb{P}^2$ (see, e.g., the textbook by Massey \cite{SurfaceClassification}). The connected sum $R \# R$ is ordinarily defined as removing an open disk from two copies of $R$ and gluing their boundaries together. For 2-dimensional surfaces, this is equivalent to gluing the disk boundaries to opposite ends of a cylindrical tube, which is the formulation we use in this construction. Thus, let $R$ be either $T^2$ or $\rr\mathbb{P}^2$. By Theorem \ref{thmCubeComplexUniversality}, we know there is some $j$ such that $R$ is homeomorphic to $\abs{C}$, the geometric realization of a cubical complex $C$ with $j$ variables. Let $\seq{x}{j}$ be these variables and write $\mathbf{x}=(\seq{x}{j})$. 
	
	Clearly $C$ has at least two $2$-dimensional faces. Let $C_1$ be a cubical complex obtained by removing a face $f_1$ of dimension $2$ from $C$, and $C_2$ obtained by removing a different face $f_2$. For $i \in \{1, 2\}$, define $B_i$ to be the cubical complex in $\mathbf{x}$ consisting of $f_i$ and its boundary, so that $|C_i|\cap |B_i|$ is the boundary of the removed disc (see Figure \ref{figConnectSumParts}).
	
	\ipnc{.7}{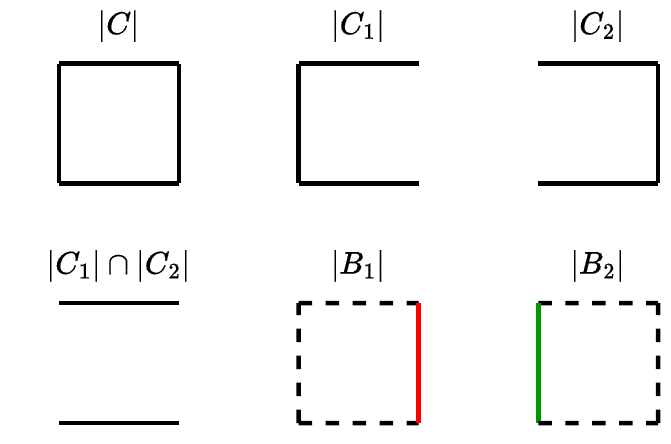}{\label{figConnectSumParts} Visualisation of $|C_1|$, $|C_2|$, $|B_1|$ and $|B_2|$. In reality, $C$ should be a complex for $T^2$ or $\rp{2}$, rather than $S^1$ as shown. The figure is just meant to convey the construction at a schematic level.}
	
	Define a new variable $x_0$ and fix constants $k_0=0<k_1<\dots<k_{g-2}<1=k_{g-1}$. Then the following formulas express the property that a point $\mathbf{x}$ is contained in $R \# R \# \dots \# R$ (see Figure \ref{figConnectSumExample} for a visualization). For $g$ even we have
	\begin{align*}
		&(\mathbf{x}\in|C_2|\vee x_0=0)\wedge (\mathbf{x}\in|C_1|\vee x_0=1)\wedge\\ &(\mathbf{x}\in|B_1|\vee 0\le x_0\le k_1\vee k_2\le x_0\le k_3\vee\dots\vee k_{g-3}\le x_0\le k_{g-2}\vee x_0=1)\wedge\\
		&(\mathbf{x}\in|B_2|\vee k_1\le x_0\le k_2\vee k_3\le x_0\le k_4\vee \dots\vee k_{g-2}\le x_0\le 1),
	\end{align*}
	and for $g$ odd we have
	\begin{align*}
		&(\mathbf{x}\in|C_2|\vee x_0=0\vee x_0=1)\wedge \mathbf{x}\in|C_1|\wedge\\
		&(\mathbf{x}\in|B_1|\vee 0\le x_0\le k_1\vee k_2\le x_0\le k_3\vee\dots\vee k_{g-2}\le x_0\le 1)\wedge\\
		&(\mathbf{x}\in|B_2|\vee k_1\le x_0\le k_2\vee k_3\le x_0\le k_4\vee \dots\vee k_{g-3}\le x_0\le k_{g-2}\vee x_0=1).
	\end{align*}
	Note that such a space is typically not the geometric realization of a cubical complex in $x_0, \seq{x}{j}$ because of the constraints on $x_0$; nevertheless, we will construct an art gallery for it.
	
	\ipnc{.7}{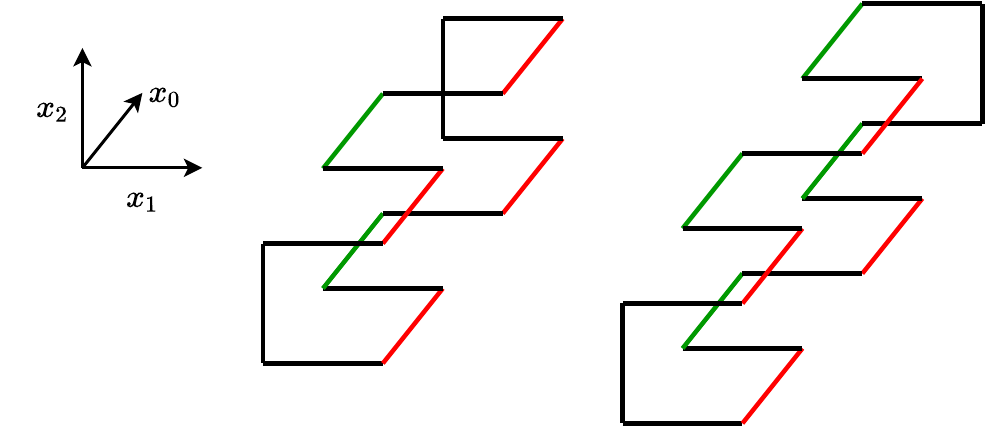}{\label{figConnectSumExample} A schematic visualization for our constructions of $R \# R \# \dots \# R$ for even $g$ (left) and odd $g$ (right). For all values of $x_0$, we require $\mathbf{x} \in |C_1|$ and $\mathbf{x} \in |C_2|$, except for $x_0=0$ or $x_0=1$ where one of these constraints is dropped to cap the hole at the end. For $k_i<x_0<k_{i+1}$, $\mathbf{x}$ must be in either $|B_1|\cap|C_1|$ or $|B_2|\cap|C_2|$ depending on the parity of $i$. This creates tubes connecting the copies of $R$.}
	
	We can write these expressions in CNF with terms of form $x_i=0$, $x_i=1$, or $k_i\le x_0\le k_{i+1}$. We write the constraint $\mathbf{x}\in |B_1|$ in CNF as
	\[\phi_{B_1} := \bigwedge_{\ell=1}^{p}\bigvee_{m=1}^{q_\ell} t_{\ell,m},\]
	where each $t_{\ell,m}$ is an atomic constraint of the form e.g. $x_i=c$. We may thus rewrite
	\begin{align*}
		&\mathbf{x}\in|B_1|\vee 0\le x_0\le k_1\vee k_2\le x_0\le k_3\vee\dots\vee k_{g-2}\le x_0\le 1\\ \equiv&\bigwedge_{\ell=1}^{p}\left(\left(\bigvee_{m=1}^{q_\ell} t_{\ell, m}\right)\vee 0\le x_0\le k_1\vee k_2\le x_0\le k_3\vee\dots\vee k_{g-2}\le x_0\le 1\right).
	\end{align*}
	
	To add the constraints on $x_0$ to each clause involving $x_0$, we must make a small modification to the construction from Section \ref{secGeneralConstruction}. Instead of having the clause gadgets along the top wall, we extend the top wall upward and put these gadgets along the right wall. It is easy to see that this does not affect our ability to create the constraints on $\mathbf{x}$. We then add the constraints on $x_0$ to each clause as in Figure \ref{figInequalityClause}; if all clause gadgets are vertically translated copies of each other, then all guard segments for $x_0$ across the different clauses can be placed at the same $x$-coordinates.
	
	\ipnc{.7}{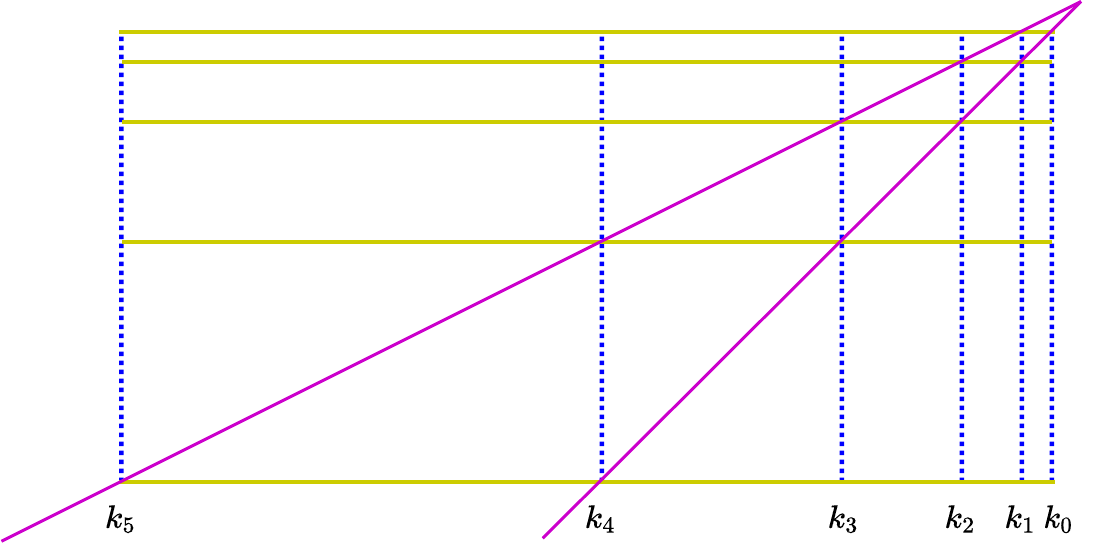}{\label{figInequalityClause} Guard segments for $x_0$. We place the first guard segment so that the lower line in the wedge intersects it a $0$, and the upper line intersects it at some point $k_1$. We can then place another segment which intersects the \emph{lower} line at $k_1$, and the intersection of this segment with the upper line gives the value of $k_2$. We can repeat until we have $g-1$ segments. Each clause depending on $x_0$ will have a copy of this setup but with some of the segments removed depending on which terms $k_i \leq x_0 \leq k_{i+1}$ it contains.}
	
	Using a similar expansion for the other terms, we obtain a CNF expression that cuts out the space we want. Since $p$ and each $q_\ell$ are constants, the total number of clauses does not depend on $g$, and each clause has at most $O(g)$ terms. Thus, when we add the constraints involving $x_0$ to the constant-sized formula $\phi_{B_1}$, the number of vertices increases linearly in $g$.
	
	We may similarly add $x_0$ constraints to the formula for $B_2$ with linear blowup. The result is a gallery with $O(g)$ vertices whose solution space is homeomorphic to $R \# R \# \dots \# R$.
\end{proof}

We have shown that such $P_g$, $Q_g$ exist having $O(g)$ vertices. In case it is of interest, we leave it to the reader to verify that the positions of the vertices can additionally be chosen to be rational numbers that require only $O(g)$ bits to describe.\footnote{In general, if we fix the length of the guard segments for $x_0$ then the wedge parameters that would give $k_{g-1}=1$ are not rational numbers. Instead, we fix a sufficiently thin slit and choose the length of the guard segments for $x_0$ appropriately.}

\section{Conclusion}\label{secConclusion}

In this work we have settled the open question of Bertschinger et al.\txt{} \cite{SOSAMain} by showing that solution spaces to the art gallery problem can capture the topology of any semi-algebraic set up to homeomorphism. In doing so, we have introduced a new form of copying gadget that enables simpler arguments about the structure of valid solutions to art gallery problem instances.

Beyond the art gallery problem, our main result raises intriguing possibilities for the broader theory of $\exists\rr$-hardness. To the best of our knowledge, this is the first paper showing that the topological structure of semi-algebraic sets can be carried into a different problem domain even up to the fine-grain notion of homeomorphism.\footnote{We are aware of one similar work by Dobbins, Holmsen, and Miltzow \cite{NPP} for the Nested Polytope Problem, which is also $\exists\rr$-complete. They show universality up to rational homeomorphism for bounded real varieties, which is a smaller class than semi-algebraic sets.} Perhaps this holds for other $\exists\rr$-hard problems as well.

\section*{Acknowledgments}

We are very grateful to Simon Weber for his helpful, detailed feedback on an earlier version of this paper. We are also grateful to Tillmann Miltzow for numerous helpful suggestions. This material is based upon work supported by the National Science Foundation Graduate Research Fellowship Program under Grant No.\txt{} DGE1745303. Any opinions, findings, and conclusions or recommendations expressed in this material are those of the authors and do not necessarily reflect the views of the National Science Foundation.

\bibliographystyle{plain}
\bibliography{bibliography}

\end{document}